\documentclass[journal]{IEEEtran}
\usepackage[utf8]{inputenc}
\usepackage[T1]{fontenc}
\usepackage{url}
\usepackage{ifthen}
\usepackage{cite}
\usepackage[cmex10]{amsmath}
\usepackage{balance}
\usepackage{amssymb,bm}
\usepackage{graphicx}
\usepackage{amsthm}
\usepackage{xcolor}
\newtheorem{theorem}{Theorem}
\usepackage{acronym}  
\acrodef{dl}[DL]{deep learning}
\acrodef{mccr}[MCCR]{maximum channel coding rate}
\acrodef{mbep}[MBEP]{maximum block error probability}
\acrodef{dnn}[DNN]{deep neural network}
\acrodef{sic}[SIC]{successive interference cancellation}
\acrodef{nn}[NN]{neural network}
\acrodef{pdf}[PDF]{probability density function}
\acrodef{pmf}[PMF]{probability math function}
\acrodef{sto}[STO]{symbol timing offset}
\acrodef{isi}[ISI]{intersymbol interference}
\acrodef{awgn}[AWGN]{additive white Gaussian noise}
\acrodef{iot}[IoT]{Internet of Things}
\acrodef{cdma}[CDMA]{code division multiple access}
\acrodef{bs}[BS]{base station}
\acrodef{bpsk}[BPSK]{binary phase-shift keying}
\acrodef{dsss}[DS-SS]{direct-sequence spread spectrum}
\acrodef{sssr}[SSSR]{simultaneous sparse signal reconstruction}
\acrodef{tdd}[TDD]{time-delay diversity}
\acrodef{tddssa}[TDD-SSA]{time-delay diversity sparse signal approximation}
\acrodef{kkt}[KKT]{Karush-Kuhn-Tucker}
\acrodef{map}[MAP]{maximum a posteriori probability}
\acrodef{somp}[S-OMP]{Simultaneous Orthogonal Matching Pursuits}
\acrodef{ls}[LS]{least square}
\acrodef{mud}[MUD]{multiuser detection}
\acrodef{mse}[MSE]{mean squared error}
\acrodef{ber}[BER]{bit error rate}
\acrodef{wcss}[WCSS]{within-cluster sum of squares}
\acrodef{sb}[SB]{Symbol-Based}
\acrodef{rd}[RD]{ridge regression}
\acrodef{ssr}[SSR]{sparse signal reconstruction}
\acrodef{iid}[i.i.d.]{independent and identically distributed}
\acrodef{ls}[LS]{least-squares}
\acrodef{mse}[MSE]{mean-squared error}
\acrodef{2mc}[2MC]{2-mean clustering}
\acrodef{sae}[SAe]{sparsity-aware}
\acrodef{wcss}[WCSS]{least within-cluster sum of squares}
\acrodef{pb}[PB]{Packet-Based}
\acrodef{pmf}[PMF]{probability mass function}
\acrodef{cv}[CV]{cross-validation}
\acrodef{mmse}[MMSE]{minimum mean squared error}
\acrodef{snr}[SNR]{signal-to-noise ratio}
\acrodef{cvpl}[CVPL]{cross-validated partial likelihood}
\acrodef{sdl}[SDL]{sparse dictionary learning}
\acrodef{mod}[MOD]{method of optimal directions}
\acrodef{stls}[S-TLS]{sparse-total least square}
\acrodef{ftn}[FTN]{faster-than-Nyquist}
\acrodef{dof}[DoF]{degrees of freedom}
\acrodef{psd}[PSD]{power spectral density}
\acrodef{cv}[CV]{cross validation}
\acrodef{mlr}[MLR]{maximum likelihood ratio}
\acrodef{gcv}[GCV]{generalized cross validation}
\acrodef{cdf}[CDF]{cumulative distribution function}
\acrodef{pls}[P-LS]{penalized-LS}
\acrodef{ldpc}[LDPC]{low-density parity-check}
\acrodef{dc}[DC]{Decision Combining}
\acrodef{ec}[EC]{Estimate Combining}
\acrodef{roc}[ROC]{receiver operating characteristic}
\acrodef{mtc}[MTC]{machine-type communications}
\acrodef{ma}[MA]{multiple access}
\acrodef{mac}[MAC]{media access control}
\acrodef{phy}[PHY]{physical}
\acrodef{ra}[RA]{random access}
\acrodef{fcc}[FCC]{fading channel coefficients}
\acrodef{cfo}[CFO]{carrier frequency offset}
\acrodef{ble}[BLE]{Bluetooth low-energy}
\acrodef{rfid}[RFID]{radio frequency identification}
\acrodef{csma}[CSMA]{carrier sensing multiple access}
\acrodef{ca}[CA]{collision avoidance}
\acrodef{lte}[LTE]{Long-Term Evolution}
\acrodef{rfid}[RFID]{radio frequency identification}
\acrodef{fsa}[FSA]{frame slotted ALOHA }
\acrodef{lpwa}[LPWA]{low-power wide area}
\acrodef{rftdm}[R-FTDM]{random frequency-time division multiplexing}
\acrodef{rpma}[RPMA]{random phase multiple access}
\acrodef{cdma}[CDMA]{code division multiple access}
\acrodef{lbt}[LBT]{listen-before-talk}
\acrodef{ss}[SS]{spread spectrum}
\acrodef{bch}[BCH]{Bose–Chaudhuri–Hocquenghem}
\acrodef{ap}[AP]{Access Point}
\acrodef{3gpp}[3GPP]{3rd Generation Partnership Project}
\acrodef{nb}[NB]{narrowband}
\acrodef{fdma}[FDMA]{frequency division multiple access}
\acrodef{rach}[RACH]{Random Access Channel}
\acrodef{fa}[FA]{fixed assignment}
\acrodef{tdma}[TDMA]{time-division multiple access}
\acrodef{ds}[DS]{direct sequence}
\acrodef{dsa}[DSA]{device sparsity-aware}
\acrodef{pdsa}[PDSA]{packet-device sparsity-aware}
\acrodef{pts}[PTS]{packet transmission state}
\acrodef{mld}[MLD]{maximum likelihood decoding}
\acrodef{cb}[CB]{contention-based}
\acrodef{css}[CSS]{chirp spread spectrum}
\acrodef{mf}[MF]{matched filter}
\acrodef{ltem}[LTE-M]{Long Term Evolution for Machines}
\acrodef{cp}[CP]{carrier phase}
\acrodef{per}[PER]{packet error rate}
\acrodef{dcd}[DCD]{differentially coherent decorrelation}
\acrodef{ca}[CA]{code-aided}
\acrodef{adam}[Adam]{adaptive moment estimation}
\acrodef{mimo}[MIMO]{multiple-input multiple-output}
\acrodef{nl}[NC]{nulling and cancelling}
\acrodef{spi}[SD-IRS]{sphere decoding with increasing radius search}
\acrodef{urllc}[URLLC]{ultra-reliable and low-latency communications}

\ifCLASSINFOpdf
\else
\fi

\begin{document}
\title{ \vspace{-0.45em} \LARGE{Efficient Massive Machine Type Communication (mMTC) via AMP}}
\author{Mostafa~Mohammadkarimi~\IEEEmembership{Member,~IEEE},
        and~Masoud~Ardakani,~\IEEEmembership{Senior~Member,~IEEE}

\thanks{This work was supported by the Natural Sciences
and Engineering Research Council of Canada (NSERC).}
\thanks{M. Mohammadkarimi is with
the
Department of Microelectronics at
Delft University of Technology, 2628 CD Delft, the Netherlands (e-mail: m.mohammadkarimi@tudelft.nl).}
\thanks{M. Ardakani is with the Department of Electrical and Computer Engineering at University of Alberta, Edmonton, AB T6G 2R3, Canada (e-mail:
ardakani@ualberta.ca.}
}

%
%

\markboth{{This paper published in  IEEE Wireless Communications Letters (DOI: 10.1109/LWC.2023.3256300).}}
{Shell \MakeLowercase{\textit{et al.}}: Bare Demo of IEEEtran.cls for IEEE Journals}
%



\maketitle

\begin{abstract}
We propose efficient and low-complexity \ac{mud} algorithms for Gaussian multiple access channel (G-MAC)
for short-packet transmission in massive machine type communications.
To do so, we first formulate the G-MAC MUD problem as a sparse signal recovery problem and
obtain the exact and approximate joint
prior distribution of the sparse vector to be recovered. Then, we employ
the Bayesian approximate message passing (AMP) algorithms with the optimal separable and non-separable \ac{mmse} denoisers for soft decoding of the sparse vector.
The effectiveness of the proposed MUD algorithms for
a large number of devices
is supported by simulation results.
For packets of $8$ information bits, while the state-of-the-art AMP with soft-threshold denoising achieves $8/100$ of the upper bound at $E_b/N_0=4$ dB,
the proposed algorithms reach $4/7$ and $1/2$ of the upper bound.

\end{abstract}

\begin{IEEEkeywords}
Multiuser detection, approximate message passing, IoT, Bayesian MMSE denoiser, sparse recovery, short packet.
\end{IEEEkeywords}
\acresetall
\vspace{-2.2em}
\IEEEpeerreviewmaketitle
\section{Introduction}
In many \ac{iot} applications, massive machine-type communications (mMTC) services  are  required,  where  a large number of devices  transmit  very  short  packets to a \ac{bs}. Typically, these short  packets  are transmitted whenever a measured value changes; thus, sporadic traffic pattern is popular in mMTC \cite{mohammadkarimi2018signature,mohammadkarimi2021massive,shiv2022learning}.
Sparse signal processing for
massive connectivity including data and activity detection has been studied in \cite{shiv2022learning,liu2018sparse,di2020detection}.
%

Recently, bounds on the exact block
error probability of Gaussian random codes in \ac{awgn} for massive uplink multiple access (MA) at finite message length have been derived \cite{zadik2019improved}.
The author in \cite{zadik2019improved} has shown that
a promising solution to support massive uplink connectivity for short packet transmission is the Gaussian multiple access channel (G-MAC).
The challenging task in the realization of G-MAC is
the design of efficient and low-complexity
multiuser detection (MUD) algorithms at the BS \cite{choi2019noma}.
In \cite{muller2021massive}, the iterative soft interference cancellation (SIC) decoding has been employed for the G-MAC with short packet transmission, but it suffers from error propagation.

In this paper, we propose two efficient and low-complexity MUD algorithms for the G-MAC for short-packet mMTC.
We first formulate the MUD problem as an under-determined sparse signal recovery problem.
Then, we propose to employ approximate message passing (AMP) for soft decoding of the short packets transmitted by the IoT devices in the network.
In most MUD problems using AMP for soft decoding, the joint prior distribution
of the vector to be recovered is generally assumed to be unknown.
In such cases, the optimal \ac{mmse} denoiser cannot be employed, and thus;
the denoiser is designed under the minimax framework to
optimize the AMP algorithm performance for the worst-case
or least-favorable distribution, leading to
a soft-thresholding denoiser.
However, in the problem at hand, the key idea is that we can obtain the exact and approximate prior joint distribution of the sparse vector to be recovered.  Thus,
the optimal separable and non-superable \ac{mmse} denoisers can be employed for Bayesian AMP soft decoding.
The decoupling effect results in low-complexity MMSE denoising in each AMP iteration \cite{mayer2015bayesian,ma2019approximate,rangan2011generalized}.
The contributions of the paper are:
\begin{itemize}
\item
We formulate the problem of G-MAC MUD as a sparse signal recovery problem and obtain the
exact and approximate prior joint distribution of the sparse vector to be recovered. The approximate prior distribution is obtained by minimizing the Kullback-Leibler divergence (KLD) as a measure of difference between probability densities.
\item
By taking into account the exact and the KLD-based approximate prior distribution of the sparse vector, we introduce two MUD algorithms called (i)
the relaxed block sparsity AMP (RBS-AMP), and (ii)
the block sparsity AMP (BS-AMP)  MUD algorithms for G-MAC suitable for short packet transmission in mMTC.
\item
The multiuser efficiency (MUE) of the proposed MUD algorithms is derived and shown to be close to one.
\item
We show that our BS-AMP MUD algorithm results in a new achievability bound  that is closer to the upper bound in \cite{zadik2019improved} compared to the achievability bound of the RBS-AMP, the state of the art AMP soft decoding in \cite{donoho2009message}, and the SIC decoding in \cite{muller2021massive}.
\end{itemize}

{\it Notation:}
The symbols $(\cdot)^{\rm T} $, ${\rm Tr}(\cdot)$, $\| \cdot\|_0$, and $\|\cdot\|$
show  the transpose operator, trace of a matrix, $0$-norm, and Euclidean norm of a vector, respectively.
The symbol $\hat{\bf{x}}$ is an estimate of vector $\bf x$, and
$\delta (\cdot)$ is the Dirac delta function.
\vspace{-0.5em}
\section{System Model}\label{sec:intro}
We consider an \ac{iot} network equipped with open loop power control (OLPC), where $D$ IoT devices simultaneously communicate with a single \ac{bs}. In the OPLC, there is dedicated pilot channel provided for channel estimation. Pilot signal is transmitted by the BS to all \ac{iot} devices in the network, and the \ac{iot} devices receive the pilot signal and estimate the power strength. Based on this estimate, each \ac{iot} device adjusts its transmit power accordingly to compensate the effect of channel \cite{muller2021massive}. In the uplink channel,
each IoT device transmits
$K$ bits of information by using Gaussian random coding with codewords of length $M=nD$.
The aggregate rate of the multiuser coded system is $DK / M$.
For short packet transmission, the received vector at the \ac{bs} in vector form is
\begin{align}\label{Equation_2}
{\bf y}={\bf C}{\bf x}+{\bf w},
\end{align}
where ${\bf C} \in \mathbb{R}^{M \times N}$ has
\ac{iid} entries, $C_{ij} \sim \mathcal{N} (0,1/M)$, $N  \triangleq  2^KD$, and
${\bf x}$ is a block sparse vector of length  $N\times 1$ with the sparsity level $D$ as
\begin{align}\label{eq3}
{\bf{x}}=\big{[}x_1 \ x_2 \ \ldots, x_N \big{]}^{\rm T} =\big{[}{\bf{x}}_1^{\rm T} \ {\bf{x}}_2^{\rm T} \dots \ {\bf{x}}_D^{\rm T} \big{]}^{\rm T},
\end{align}
where
${\bf{x}}_q$ is a $2^K \times 1$
vector as $|| {\bf{x}}_q||_0=1, q=1,2,\dots, D$.
The $M \times 1$ vector ${\bf w}$ in \eqref{Equation_2} is the \ac{awgn} vector with covariance matrix $\sigma_{\rm w}^2 \bf{I}$.
We assume that $ \beta \triangleq \mathop {\lim }\limits_{D \to \infty}
M/N   = n/(2^K)  >0$.


Because of the OLPC, the received power of the $D$ \ac{iot} devices at the \ac{bs} is fixed and is denoted by $P_{\rm{I}}$. In this case,
the joint prior distribution of the transmit
vector ${\bf x}_q \triangleq [x_{q,1},x_{q,2},\dots, x_{q,2^{K}}]^{\rm T}$, $q=1,2,\dots, D$, in \eqref{eq3} is given by\footnote{The corresponding probability mass function (PMF)  of the PDF  in \eqref{qpdf} is $P\{{\bf x}_q=\sqrt{{P}_{\rm I}}{\bf e}_k\}=\frac{1}{2^K}$, $k=1,2,\ldots,2^K$.
Since for the problem at hand, we have both continues and discrete random variables,  Dirac’s delta function is used to link discrete PMF to PDF.}
\begin{align}\label{qpdf}
f_{\bf X}({\bf x}_q;P_{\rm{I}},K)= \sum_{k=1}^{2^K} \frac{1}{2^K}\delta({\bf x}_q -\sqrt{P_{\rm{I}}}{\bf e}_k),
\end{align}
where ${\bf e}_1,{\bf e}_2, \cdots, {\bf e}_{2^K}$ are orthonormal basis in right-handed Cartesian coordinate system. 
Our goal is to design MUD algorithms to detect the $K$ information bits transmitted  by each \ac{iot} device.
As ${\bf x}$ is a block-sparse vector, we have a sparse signal reconstruction problem.
In this paper, we
adopt the low-complexity Bayesian AMP algorithms for signal recovery.
\vspace{-0.2em}
\section{Review of AMP}
In this section, we first briefly review AMP
and its state evolution (SE) for the separable denoiser.
%
\subsection{AMP Algorithm with Separable Denoiser}
For noisy linear measurements of the form
\begin{align}\label{Equation_5}
{\bf s}={\bf A}{\bf u}+{\bf w},
\end{align}
where ${\bf A} \in \mathbb{R}^{M \times N}$, ${\bf s} \in \mathbb{R}^{M \times 1}$, ${\bf u} \in \mathbb{R}^{N \times 1}$, and ${\bf w} \in \mathbb{R}^{M \times 1}$.
AMP algorithm with denoiser ${\bm{\eta}}_t$: $\mathbb{R}^N\rightarrow \mathbb{R}^N$ at the $(t+1)$th iteration estimates vector $\bf{u}$ as follows \cite{berthier2020state}
\begin{subequations}
\begin{align}\label{Equation_6}
\hat{\bf u}^{t+1}&= {\bm{\eta}}_t ({\bf v}^{t})= {\bm{\eta}}_t (\hat{\bf{u}}^{t} + {\bf A}^{\rm T} {\bf r}^{t}), \\ \label{Equation_7}
{\bf r}^t&={\bf s}-{\bf A}\hat{\bf{u}}^t+\frac{N}{M}{\bf r}^{t-1}{\rm div}\big{(}{\bm{\eta}}_{t-1} ({\bf v}^{t-1})\big{)},  \\
 {\bf v}^t&= \hat{\bf{u}}^t + {\bf A}^{\rm T} {\bf r}^t, \label{Equation_8}
\end{align}
\end{subequations}
where ${\rm div}\big{(}{\bm{\eta}}_t ({\bf g})\big{)} \triangleq \frac{1}{N} {\rm Tr} \Big{(}\frac{\partial{\bm{\eta}}_t({\bf g})}{\partial {\bf g}} \Big{)}$
with ${\bf{g}} \triangleq [g_1 \ g_2 \dots g_N]^{\rm T}$,
$\hat{\bf u}^0=\bf{0}$, and ${\bf  r}^{0}=\bf{s}$.
For identical separable Lipschitz  continuous (LC) denoiser, ${\bm{\eta}}_t ({\bf g})=\big{[}{\eta}_t ({g_1}) \ {\eta}_t ({g_2}) \ \dots \ {\eta}_t ({g_N})\big{]}^{\rm T}$,
and ${\rm div}\big{(}{\bm{\eta}}_t ({\bf g})\big{)}= \frac{1}{N} \sum_{j=1}^{N} {\partial {\eta}_t (g_j)}/{\partial g_j}.$
Similarly, for non-separable LC denoisers, we have \cite{berthier2020state}
\begin{align}\label{equation_13}
{\bm{\eta}}_t ({\bf g})=\big{[}{\bm \eta}_{t,1} ({{\bf g}_1}) \ {\bm \eta}_{t,2} ({{\bf g}_2}) \ \dots \ {\bm \eta}_{t,L} ({{\bf g}_L})\big{]}^{\rm T},
\end{align}
and
\begin{align}\label{Equation_14}
{\rm div}\big{(}{\bm{\eta}}_t ({\bf g})\big{)}= \frac{1}{N} \sum_{m=1}^{L}{\rm Tr} \Big{(}\frac{\partial{\bm{\eta}}_{t,m}({\bf g}_m)}{\partial {\bf g}_m} \Big{)},
\end{align}
where ${{\bf g}_1}, {{\bf g}_2}, \cdots, {{\bf g}_L}$ denote $L$ non-overlapping partitions of ${\bf g}$.
The initial work on AMP focused on the case where $\bm{\eta}_t(\cdot)$  is
a separable elementwise function with identical components \cite{donoho2009message}, while the later work
allowed non-separable  denoiser \cite{berthier2020state}.

For large, i.i.d., sub-Gaussian random matrices $\bf A$ in \eqref{Equation_5},\footnote{${E}\{A_{ij}\}=0$, ${E}\{A_{ij}^2\}=\frac{1}{M}$, and ${E}\{A_{ij}^6\}=\frac{C}{M}$, for some fixed $C>0$.}  $\lim_{M,N\to\infty} M/N  \triangleq \beta  >0$, and
identical separable/non-separable LC denoiser, AMP displays decoupling behavior, meaning that ${\bf v}^t$ in \eqref{Equation_8} behaves like an \ac{awgn} corrupted version of the true signal $\bf{u}$ in \eqref{Equation_5} as \cite{rush2018finite}
\begin{align}\label{Decouple}
{\bf v}^t= {\bf{u}} + {\cal{N}}(0,\sigma_t^2 {\bf I}),
\end{align}
where an estimate of $\sigma_t^2$ can be obtained as $\hat{\sigma}_t^2 = \frac{1}{M} \|{\bf{r}}^t\|$.
The larger $M$, the lower the estimation error.
\vspace{-1em}
\subsection{SE of AMP with Separable Denoiser}
For large, \ac{iid}, sub-Gaussian random matrices $\bf A$, the
performance of the AMP can be exactly predicted by a scalar SE, which also provides
testable conditions for optimality.
For identical separable LC denoiser $\eta_t(\cdot)$, when the elements of ${\bf{u}}=[u_1, u_2, \dots , u_N]^{\rm T}$
empirically converges to some random variable $U$
with \ac{pdf}
$f_{U}$, the SE i.e., the sequence of $\{\sigma_t^2\}_{t \geq 0}$ in \eqref{Decouple} is expressed as \cite{bayati2011dynamics}
\begin{subequations}\label{yujkkk}
\begin{align}
\sigma^2_{t+1} &= \sigma_{\rm w}^2+\frac{1}{\beta}\mathbb{E}\big{\{}(\eta_t(U+\sigma_t Z)-U)^2\big{\}}, \label{Equation_18}
\\
\sigma^2_{0} &= \sigma_{\rm w}^2+\frac{1}{\beta}\mathbb{E}\big{\{}U^2\big{\}}, \label{Equation_19}
\end{align}
\end{subequations}
where $\beta\triangleq M/N$, $Z \sim {\cal  N}(0,1)$, and is independent from $U \sim f_U$.
The term $(\sigma^2_{t+1} - \sigma_{\rm w}^2)\beta$ in \eqref{Equation_18} can be estimated by the sample mean estimate of the denoiser  error as \cite{bayati2011dynamics} 
\begin{align}\label{Equation_23}
 \mathop {\lim }\limits_{N \to \infty } \frac{1}{N}  \big{\|}\hat{\bf{u}}^{t+1}-{\bf{u}}\big{\|}_2^2 & = \mathbb{E}\big{\{}(\eta_t(U+\sigma_t Z)-U)^2\big{\}} \\ \nonumber
& =(\sigma^2_{t+1} - \sigma_{\rm w}^2)\beta,
\end{align}
where $\hat{\bf u}^{t+1} \triangleq [\hat{u}_1^{t+1} \ \hat{u}_2^{t+1} \ \dots \ \hat{u}_N^{t+1}]^{\rm T}$ is obtained by using the AMP algorithm in \eqref{Equation_6}-\eqref{Equation_8}.
\vspace{-0.5em}
\section{RBS-AMP MUD Algorithm}
In this section, the separable MMSE  denoiser for
the approximate joint prior distribution of the sparse vector ${\bf x}$ is obtained and then is used for Bayesian AMP soft decoding.
This results in RBS-AMP MUD algorithm.
The MUE of the proposed RBS-AMP MUD algorithm is also derived.
\vspace{-1em}
\subsection{RBS-AMP MUD Algorithm with Separate Denoiser}
The RBS-AMP MUD algorithm
neglects the block-sparsity of $\bf{x}$ in \eqref{Equation_2} and considers an \ac{iid} distribution for its elements.
Let $\eta({\bf x}_q) \triangleq {\prod_{k = 1}^{2^K} f_{X}(x_{q,k};P_{\rm{I}};K)}$ denote an approximation of the joint prior distribution ${\gamma}({\bf x}_q)\triangleq {f_{\rm {\bf X}}({{\bf{x}}_q};{P_{\rm{I}}},K)}$ in \eqref{qpdf}.
By minimizing the KLD between the corresponding PMFs of $\gamma$ and $\eta$, we obtain
\begin{align}\label{Eq_PDF_DELTA}
f_{X}(x;P_{\rm{I}};K) = \frac{1}{2^K} \delta\big{(}x-\sqrt{P_{\rm {I}}}\big{)}+\Big{(}1-\frac{1}{2^K}\Big{)}\delta(x).
\end{align}
This block sparsity relaxation
results in a low complexity MUD algorithm because \eqref{qpdf} is replaced with the simple joint prior distribution $\eta({\bf x}_q)$. While  $\eta({\bf x}_q)$ allows the sparsity level of ${\bf{x}}_q$ to be larger than one, the separable MMSE denoiser can correctly decode the one-hot vector because sparsity level one still has the highest prior probability.

\setcounter{equation}{15}
\begin{figure*}
\vspace{-1.5em}
\begin{align}\nonumber
 \mathop {\lim }\limits_{N \to \infty } \frac{1}{N}  \big{\|}\hat{\bf{x}}^{t+1}-{\bf{x}}\big{\|}_2^2 &= \mathbb{E}\big{\{}(\eta_t(X+\sigma_t Z)-X)^2\big{\}}=
\frac{1}{2^K \sqrt{2\pi\sigma_t^2}}\int_{-\infty}^{\infty}
\Bigg{(}\frac{\sqrt{P_{\rm I}} pT_i^t}{pT_i^t+1-p}-\sqrt{P_{\rm{I}}}\Bigg{)}^2
\exp \Big{\{}-\frac{(v_i^t-\sqrt{P_{\rm{I}}})^2}{2\sigma_t^2} \Big{\}} {\rm d} v_i^t \\ \label{Equation_37}
&+
\Big{(}1-\frac{1}{2^K} \Big{)}\frac{1}{\sqrt{2\pi\sigma_t^2}}\int_{-\infty}^{\infty}
\Bigg{(}\frac{\sqrt{P_{\rm I}}pT_i^t}{pT_i^t+1-p}\Bigg{)}^2
\exp \Big{\{}-\frac{(v_i^t)^2}{2\sigma_t^2} \Big{\}} {\rm d} v_i^t = (\sigma^2_{t+1} - \sigma_{\rm w}^2)\beta.
\end{align}
\noindent\rule{19cm}{0.4pt}
\vspace{-2.2em}
\end{figure*}
Let $\hat{\bf{x}}=[\hat{x}_1^{t+1},\hat{x}_2^{t+1},\cdots,\hat{x}_N^{t+1}]^{\rm T}$ denote the reconstructed vector  by the AMP algorithm summarized in \eqref{Equation_6}-\eqref{Equation_8} at the $(t+1)$th iteration for the observation model in \eqref{Equation_2}, i.e., ${\bf s}={\bf y}$, ${\bf u}={\bf x}$, and
${\bf A}={\bf C}$. For the identical separable MMSE denoiser $\eta(v_i^t)=\mathbb{E}\{\cdot|v_i^t\}$ with prior distribution $U=X \sim f_{X}(x;P_{\rm{I}};K)$ in \eqref{Eq_PDF_DELTA}, by taking into account the decoupling behavior in \eqref{Decouple}, the $i$th element of the soft decoded vector at the $(t+1)$th iteration is given by
\setcounter{equation}{11}
\begin{align}\label{Equation_34}
\hat{x}_i^{t+1}&=\eta({{v_i^t}}) = \mathbb{E}\big{\{}X_i|X_i+\sigma_t Z=v_i^t \big{\}}
\\ \nonumber
&=
\frac{\sqrt{P_{\rm I}}pT_i^t}{pT_i^t+1-p},
\end{align}
where $Z \sim {\cal  N}(0,1)$, $T_i^t \triangleq {\exp\Big{\{}\frac{-P_{\rm{I}}+2v_i^t{\sqrt{P_{\rm{I}}}}}{2\sigma_t^2}\Big{\}}}$,
and $p=1 / {2^K}$.
Proof of the general case for the prior \ac{pdf} in \eqref{qpdf} is given in Appendix \ref{Appx_2}.
Moreover, the divergence term at the $(t+1)$th iteration of the RBS-AMP MUD in \eqref{Equation_7} is obtained as
\vspace{-0.5em}
\begin{align}\label{Equation_36}
{\rm div}\big{(}{\bm{\eta}}_t ({\bf v}^{t})\big{)}&=\frac{1}{N}\sum_{i=1}^{N} \frac{\partial}{\partial v_i^t}\eta_t({{v_i^t}})
\\ \nonumber
&=\frac{1}{N}\sum_{i=1}^{N} \frac{\partial}{\partial v_i^t}\mathbb{E}\big{\{}X_i|X_i+\sigma_t Z=v_i^t \big{\}} \\ \nonumber
&=\frac{p(1-p){P_{\rm{I}}}}{\sigma_t^2 N}\sum_{i=1}^{N} \frac{T_i^t }{(pT_i^t+1-p)^2}.
\vspace{-5em}
\end{align}
\setcounter{equation}{13}
\vspace{-2em}
\subsection{MUE of the RBS-AMP MUD algorithm}
The degradation in bit error rate due to the presence of MA interference in AWGN channel can be measured by the multiuser asymptotic efficiency, which is defined as
\begin{align}\label{mue}
\xi =\frac{\sigma_{\rm w}^2}{\sigma_{\rm w}^2+I},
\end{align}
where $\sigma_{\rm w}^2$ is the variance of noise, and $I$ is the interference after soft decoding.
\begin{theorem}\label{Theorem_3}
The  MUE of the  RBS-AMP MUD algorithm at the $(t+1)$th iteration is given by
\begin{align}\label{7883}
\xi^{t+1}_{\rm RB}=\frac{\sigma_{\rm w}^2}{\sigma_{\rm w}^2+(\sigma^2_{t+1} - \sigma_{\rm w}^2)}=\frac{\sigma_{\rm w}^2}{\sigma^2_{t+1}},
\end{align}
where $\xi^{t+1}_{\rm RB} \in [0,1]$, and $\sigma^2_{t+1}$ is given in \eqref{Equation_18}-\eqref{Equation_19} for $U=X$ in \eqref{Eq_PDF_DELTA} (Proof in Appendix \ref{Appx_3}).
\end{theorem}

In practice, the SE, i.e., $\{\sigma_t^2\}_{t \geq 0}$ is tracked by the sample mean estimate $\hat{\sigma}_t^2 \approx \frac{1}{M} \|{\bf{r}}^t\|$, where ${\bf{r}}^t$ is given in \eqref{Equation_7} for ${\bf{s}}={\bf{y}}$, ${\bf{A}}={\bf{C}}$ and $\hat{\bf{u}}=\hat{\bf{x}}$.
The SE can also be theoretically obtained by using \eqref{Equation_23}. In this case,
from \eqref{Equation_23} and \eqref{Eq_PDF_DELTA}, we obtain \eqref{Equation_37} at the top of this page.
\section{BS-AMP MUD Algorithm}
In this section,
we first in Theorem 2 show  that the per-iteration mean-squared error of the reconstructed
vector through AMP with non-separable denoiser
can be predicted using a state evolution (SE). Then, the non-separable denoiser for
the exact joint prior distribution of the sparse vector ${\bf x}$ in \eqref{qpdf} is obtained and is used for Bayesian AMP soft decoding.
This results in the BS-AMP MUD algorithm.
Theorem 2 will
be later used for the derivation of the MUE of the proposed
BS-AMP MUD algorithm.
\vspace{-1em}
\subsection{SE of AMP with Non-separable Denoiser}
Let us consider a block structure for the vector ${\bf{u}}$ in \eqref{Equation_5} as ${\bf{u}}=[{\bf{u}}_1^{\rm {T}},{\bf{u}}_2^{\rm {T}},\cdots,{\bf{u}}_L^{\rm {T}}]$, where ${\bf{u}}_i$ and ${\bf{u}}_j$, $i\neq j$ are independent random vectors. We consider the following cases:
1) the joint \ac{pdf} of ${\bf{u}}_i \triangleq [u_{i,1},u_{i,2},\dots,u_{i,\ell}]^{\rm T}$ is ${\bf{U}}=[U_1 \ U_2 \ \dots \ U_{\ell} ]^{\rm T} \sim f_{\bf{U}}$, and 2) the empirical joint
\ac{pdf} of ${\bf{u}}_i$ converges weakly to $f_{\bf U}$, where $\ell=N/L$.
AMP with identical non-separable LC denoiser in \eqref{equation_13} with $\tilde{\bm \eta}_t(\cdot): \mathbb{R}^\ell \rightarrow \mathbb{R}^\ell$ $({\bm \eta}_{t,m}(\cdot)=\tilde{\bm \eta}_t(\cdot), m=1,2,\dots,L)$, reconstructs the vector ${\bf{u}}$ from the observation model in \eqref{Equation_5} through the relations in
\setcounter{equation}{16}
\eqref{Equation_6}-\eqref{Equation_8}. Also, the decoupling behavior occurs for the $m$th block as
\begin{align}\label{ddnlawqm}
{\bf v}_m^t= {\bf{u}}_m + {\bf z}_m^t,\,\,\,\,\,\,\ m=1,2,\dots,L,
\end{align}
where
\begin{align}\label{SEV}
{\bf z}_m^t \sim N\big{(}{{0}},\rho_t^2{\bf {I}}\big{)}.
\end{align}
The SE i.e., the sequence of $\{\rho_t^2\}_{t \geq 0}$ in \eqref{SEV} is given by
\begin{subequations}\label{eerefd157766}
\begin{align}
\rho^2_{t+1} &= \sigma_{\rm w}^2+\frac{1}{\ell \beta}\mathbb{E}\Big{\{}\big{\|}\tilde{\bm{\eta}}_t({\bf{U}}+\rho_t {\bf{Z}})-{\bf U} \big{\|}^2 \Big{\}},  \label{Equation_27}
\\
\rho^2_{0} &= \sigma_{\rm w}^2+\frac{1}{\ell \beta}\mathbb{E}\big{\{}{\|}{\bf U}{\|}^2\big{\}},  \label{Equation_28}
\end{align}
\end{subequations}
where ${\bf Z}=[Z_1 \ Z_2 \ \dots \ Z_\ell]^{\rm T} \sim N({{0}},{\bf {I}})$ is independent of ${\bf U}=[U_1 \ U_2 \ \dots \ U_\ell]^{\rm T}\sim P_{\bf{U}}$.
A modified version of Theorem~1 for identical non-separable LC denoiser is given as follows:
\begin{theorem}\label{Theorem_2}
Let ${\bf{A}}\in \mathbb{R}^{M \times N}$ be a matrix
with \ac{iid} entries  $A_{ij} \sim \mathcal{N}  (0,1/M)$, and assume $M/N\rightarrow \beta >0$.
Consider
a block sparse signal ${\bf{u}}=\big{[}{\bf{u}}_1^{\rm T} \ {\bf{u}}_2^{\rm T} \dots \ {\bf{u}}_L^{\rm T}\big{]}^{\rm T}$, where the joint distribution
of each block of length $\ell$ is $f_{\bf U}$ (or the empirical joint
\ac{pdf} converges weakly to $f_{\bf U}$) on $\mathbb{R}^{\ell}$.
Then, we have
\begin{align}\label{Equation_29}
\hspace{-1em} \mathop {\lim }\limits_{L \to \infty } \hspace{-0.21em} \frac{1}{L} \sum_{u=1}^{L} \big{\|}\hat{\bf u}_u^{t+1}-{\bf u}_u\big{\|}^2  \hspace{-0.15em} = \hspace{-0.15em} \mathbb{E}\Big{\{}\big{\|}\tilde{\bm{\eta}}_t({\bf{U}}+\rho_t {\bf{Z}})-{\bf U}\big{\|}^2 \Big{\}},
\end{align}
where $\hat {\bf{u}}=\big{[}\hat {\bf{u}}_1^{\rm T} \ \hat {\bf{u}}_2^{\rm T} \dots \ \hat {\bf{u}}_L^{\rm T}\big{]}^{\rm T}$ is obtained by using \eqref{Equation_6}-\eqref{Equation_8},  ${\bf U} \sim f_{\bf U}$, ${\bf Z} \sim N({{0}},{\bf {I}})$ is independent of ${\bf U}$, and $\rho_t$ follows the  SE in \eqref{Equation_27} and \eqref{Equation_28} (Proof in Appendix \ref{Appx_1}).
\end{theorem}

\setcounter{equation}{22}
\begin{figure*}
\vspace{-1.5em}
\begin{align}\label{wenm12ereertw}
\hspace{-2.2em}
{\rm div}& \big{(}{\bm{\eta}}_t ({\bf v}^{t})\big{)}=
\frac{1}{N} \sum_{u=1}^{D}
 {\rm Tr} \dfrac{\partial {\tilde{\bm{\eta}}_t ({\bf v}_u^{t})}}{\partial {\bf v}_u^{t}}
=\frac{\sqrt{P_{\rm{I}}}}{N}\sum_{u=1}^{D} \frac{1}{{\sum_{n=1}^{2^K} \exp\Big{\{}\frac{-\|{\bf v}_u^t-\sqrt{P_{\rm I}}{\bf{e}}_n\|^2_2}{2\rho_t^2}\Big{\}}}}  \sum_{k=1}^{2^K}
{\frac{(\sqrt{P_{\rm{I}}}-v_{u,k})}{\rho_t^2}\exp\Big{\{}\frac{-\|{\bf v}_u^t-\sqrt{P_{\rm I}}{\bf{e}}_k\|^2_2}{2\rho_t^2}\Big{\}}}
\\ \nonumber
&\hspace{-2em} -\frac{\sqrt{P_{\rm{I}}}}{N}\sum_{u=1}^{D} \frac{1}{\Big{[}{\sum_{n=1}^{2^K}
\exp\Big{\{}\frac{-\|{\bf v}_u^t-\sqrt{P_{\rm I}}{\bf{e}}_n\|^2_2}{2\rho_t^2}\Big{\}}}\Big{]}^2}   {{\sum_{k=1}^{2^K}
\exp\Big{\{}\frac{-\|{\bf v}_u^t-\sqrt{P_{\rm I}}{\bf{e}}_k\|^2_2}{2 \rho_t^2}\Big{\}}
\sum_{m=1}^{2^K} \frac{(\sqrt{P_{\rm I}} \delta[k-m]-v_{u,k})}{\rho_t^2}
\exp\Big{\{}\frac{-\|{\bf v}_u^t-\sqrt{P_{\rm I}}{\bf{e}}_m\|^2_2}{2 \rho_t^2}\Big{\}}}}.
\end{align}
\noindent
\rule{19cm}{0.4pt}
\vspace{-2.2em}
\end{figure*}

\begin{figure}[!h]
\vspace{-1.1em}
\centering
    \includegraphics[height=1.8in]{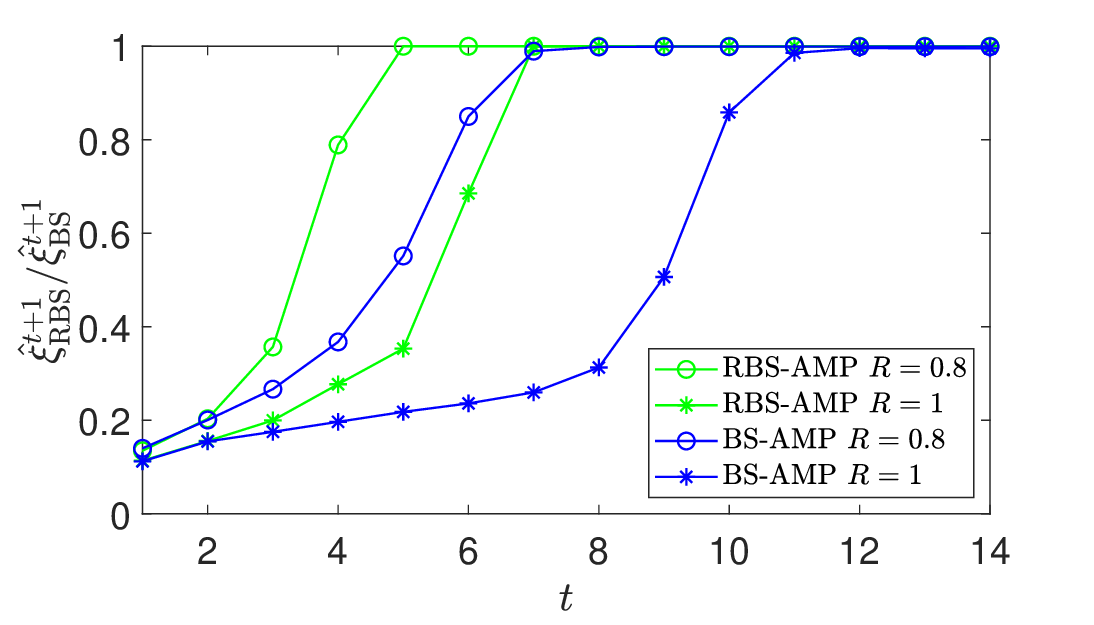}
    \vspace{-0.75em}
    \caption{State evolution of the RBS-AMP and BS-AMP MUD algorithms for $K=8$ and $D=256$ at ${E_{\rm b}}  / {N_0}=6$ dB.} \label{fig2}
\end{figure}
\begin{figure}[!h]
 \vspace{-1.2em}
\centering
    \includegraphics[height=2.5in]{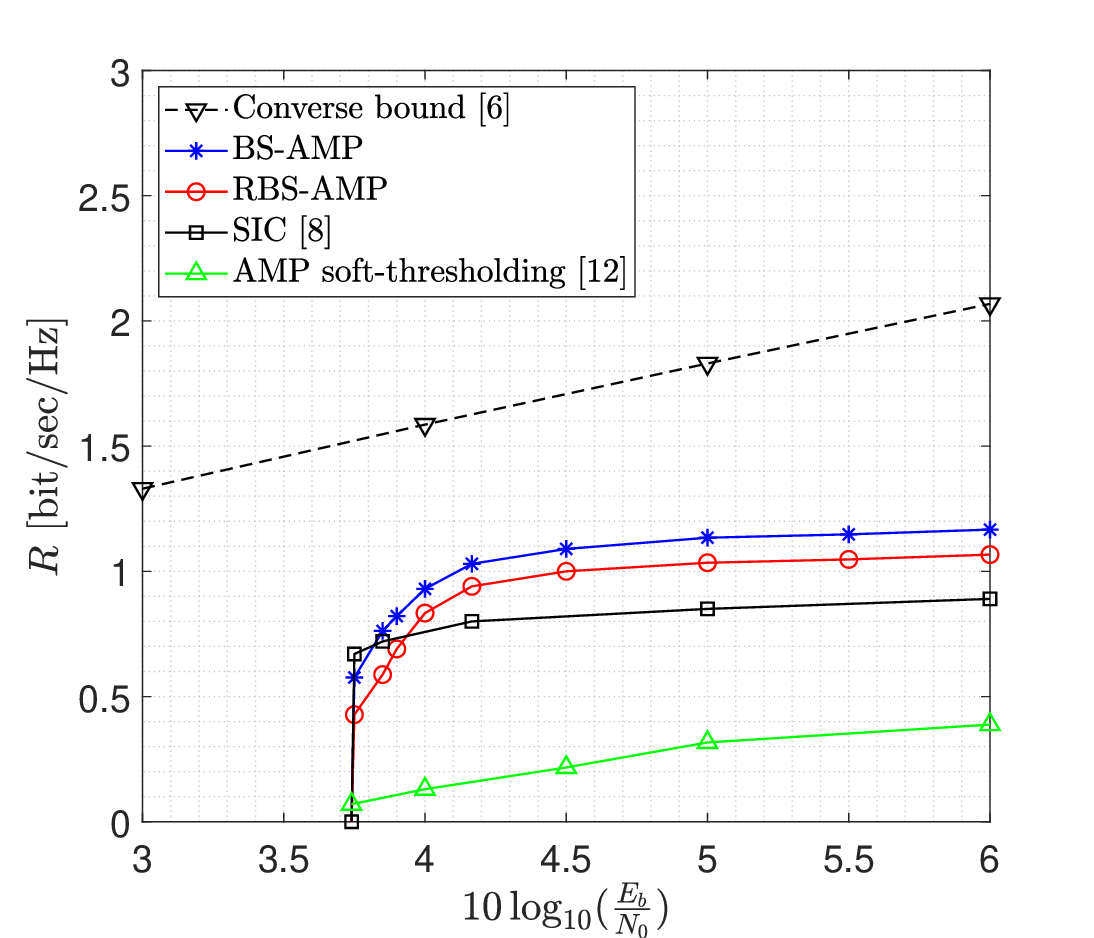}
    \caption{Spectral efficiency versus $10 \log_{10} ({E_{\rm b}}  / {N_0})$
 for the average block error rate $P_{\rm{e}}=15 \times 10^{-4}$, $K=8$, and $D=256$.} \label{F2}
 \vspace{-1.6em}
\end{figure}
\setcounter{equation}{20}
\vspace{-1em}
\subsection{BS-AMP MUD Algorithm with Non-separable Denoiser}
By considering the exact  prior distribution of
$\bf{x}$ in \eqref{qpdf},
 AMP with block non-seperable denoiser
can be employed to reconstruct the vector ${\bf x}=[{\bf x}_1^{\rm T}, {\bf x}_2^{\rm T}, \dots,{\bf x}_{D}^{\rm T}]^{\rm T}$ from the underdetermined system of linear equations ${\bf y}={\bf C}{\bf x}+{\bf w}$ in
\eqref{Equation_2} by using the relations \eqref{Equation_6}-\eqref{Equation_8}. For identical block non-separable \ac{mmse} denoiser $\tilde{\bm \eta}({\bf v}_u^t)= \mathbb{E}\{\cdot|{\bf v}_u^t\} :$ $\mathbb{R}^{{2^{K}}} \rightarrow \mathbb{R}^{{2^{K}}}$, the reconstructed vector for the $u$th \ac{iot} device at the $(t+1)$th iteration is given as follows
\begin{align}\label{bay}
\hat{\bf x}_u^{t+1}=\mathbb{E}\big{\{}{\bf{X}}|{\bf{X}}+\rho_t {\bf{Z}}={\bf v}_u^t \big{\}},
\end{align}
where ${\bf{X}}=[X_1, X_2, \cdots, X_{2^K}]^{\rm T}$ is a random vector with the joint \ac{pdf} in \eqref{qpdf} and ${\bf v}_u^t= \hat{\bf{x}}_u^t + {\bf C}^{\rm T}{\bf r}^t={\bf{x}}_u + {\cal{N}}(0,\rho_t^2 {\bf I})$ due to the decoupling affect as it was explained in \eqref{ddnlawqm} in \eqref{SEV}.

By employing \eqref{qpdf} and the fact that ${\bf Z} \sim \mathcal{N}({\bf{0}},{\bf I})$, $\hat{\bf x}_u^{t+1}$ in \eqref{bay} is obtained as
\begin{align}\label{eq25nv}
\hat{\bf x}_u^{t+1} & =\tilde{\bm{\eta}} ({\bf v}_u^{t})=\mathbb{E}\big{\{}{\bf{X}}|{\bf{X}}+\rho_t {\bf{Z}}={\bf v}_u^t \big{\}} \\ \nonumber
&=\frac{\sqrt{P_{\rm{I}}}\sum_{k=1}^{2^K} \exp\Big{\{}\frac{-\|{\bf v}_u^t-\sqrt{P_{\rm I}}{\bf{e}}_k\|^2_2}{2\rho_t^2}\Big{\}}{\bf e}_k}
{\sum_{m=1}^{2^K} \exp\Big{\{}\frac{-\|{\bf v}_u^t-\sqrt{P_{\rm I}}{\bf{e}}_m\|^2_2}{2\rho_t^2}\Big{\}}},
\end{align}
where ${\bf e}_1,{\bf e}_2, \cdots, {\bf e}_{2^K}$ are orthonormal basis 
(Proof in Appendix \ref{Appx_2}).

By using \eqref{Equation_14} and \eqref{eq25nv} and after some simplifications, the divergence for the MMSE non-separable denoiser is derived as in \eqref{wenm12ereertw}.
The RBS-AMP MUD algorithm offers a lower complexity compared to the BS-AMP MUD algorithm due to the simple separable \ac{mmse} denoiser. On the other hand, the BS-AMP MUD algorithm achieves a higher spectral efficiency compared to the RBS-AMP MUD algorithm because of the non-separable \ac{mmse} denoiser
that uses the exact joint \ac{pdf}.
\vspace{-1em}
{\subsection{MUE of the BS-AMP MUD algorithm}}
At the $(t+1)$th iteration of the BS-AMP algorithm,
 we have $\hat{\bf x}^{t+1} \triangleq [\hat{\bf x}_1^{t+1} \ \hat{\bf x}_2^{t+1} \ \dots \ \hat{\bf x}_{D}^{t+1}]^{\rm T}$, where $\hat{\bf x}_u^{t+1}$ is given \eqref{eq25nv}. Similar to the RBS-AMP algorithm in Appendix \ref{Appx_3}, we can obtain the MUE for the BS-AMP algorithm.
For $D\rightarrow \infty$, by taking into account Theorem 2 for $L=D$ and $\ell=2^K$, using the fact that ${\bf{C}}^{\rm T} {\bf{C}}={\bf I}$ ($N\rightarrow \infty$, $M\rightarrow \infty$), and employing \eqref{Equation_27} and \eqref{Equation_29}, we can write the  interference term as $I^{t+1}
=(\rho^2_{t+1} - \sigma_{\rm w}^2)$,
where $\rho^2_{t+1} = \sigma_{\rm w}^2+\frac{1}{2^K\beta}\mathbb{E}\big{\{}\|\tilde{\bm{\eta}}({\bf{X}}+\rho_t {\bf{Z}})-{\bf X}\|^2 \big{\}}$
and $\rho^2_{0} = \sigma_{\rm w}^2+\frac{1}{2^K \beta}\mathbb{E}\{\|{\bf X}\|^2\}$.
By substituting $I^{t+1}$
 into
\eqref{mue}, we obtain the MUE of the  BS-AMP MUD algorithm at the $(t+1)$th iteration as follows
\setcounter{equation}{23}
\begin{align}
\xi_{\rm BA}^{t+1}=\frac{\sigma_{\rm w}^2}{\sigma_{\rm w}^2+(\rho^2_{t+1} - \sigma_{\rm w}^2)}=\frac{\sigma_{\rm w}^2}{\rho^2_{t+1}}.
\end{align}
\section{Simulation Results}
In practice, the sample mean estimate of $\{\sigma_t^2\}_{t \geq 0}$ and $\{\rho_t^2\}_{t \geq 0}$, i,e., $\frac{1}{M} \|{\bf{r}}^t\|$ is employed to track the MUE of the  RBS-AMP and the BS-AMP MUD algorithms as $\frac{M\sigma_{\rm w}^2}{\|{\bf{r}}^{t+1}\|}$.
In Fig.~\ref{fig2}, we show the estimated MUE for $K=8$ and $D=256$ at ${E_{\rm b}}  / {N_0}=6$ dB.
As seen, $\hat{\xi}_{\rm RB}^{t+1}=1$ and $\hat{\xi}_{\rm BA}^{t+1}=1$  are achieved. Also, the lower the spectral efficiency $R$ the faster the decoding.

In Fig.~\ref{F2}, we show the spectral efficiency $R$ [bits/sec/Hz] versus $10 \log_{10} ({E_{\rm b}}  / {N_0})$
 for average block error rate $P_{\rm{e}}=15 \times 10^{-4}$, $K=8$, and $D=256$.
 We  also illustrate the spectral efficiency of the AMP with soft-thresholding \cite{donoho2009message}, the SIC decoding in \cite{muller2021massive}, and the upper bound  in \cite{zadik2019improved}.
As expected, the BS-AMP outperforms the RBS-AMP since the non-sparable MMSE denoiser in more efficiently suppresses noise.
Moreover, as seen, both algorithms outperform the AMP with soft-thresholding and the SIC. For packets of $8$ information bits,
the BS-AMP and RBS-AMP MUD algorithms reach $4/7$ and $1/2$ of the upper bound at $E_b/N_0=4$ dB, respectively.
\vspace{-0.2em}
\section{Conclusion}
Two new MUD algorithms for massive G-MAC with short-packet transmission were developed in this letter.
Our proposed MUD algorithms were developed based on AMP with non-separable and separable MMSE denoisers.
Our results show that the proposed MUD algorithms
offer superior soft decoding performance compared to the state-of-the-art AMP with soft-threshold denoising.
This higher spectral efficiency is achieved with low computational complexity for  massive G-MAC with short-packet transmission due to  decoupling effect.
\appendices
\vspace{-0.5em}
\section{}\label{Appx_2}
For the proposed BS-AMP MUD with block non-separable \ac{mmse}  denoiser with the prior \ac{pdf} in \eqref{qpdf}, we have
\begin{align} \nonumber
\hat{\bf x}_u^{t+1}&=\mathbb{E}\big{\{}{\bf{X}}|{\bf{X}}+\rho_t {\bf{Z}}={\bf v}_u^t \big{\}}
 = \int_{\mathcal{R}^{2^K}} \hspace{-0.3em} {\bf{x}}_u f_{{\bf X}|{\bf V}}( {\bf x}_u| {\bf v}_u^t ;K) \  {\rm d}  {\bf{x}}_u
\\ \nonumber
&  = \int_{\mathcal{R}^{2^K}} {\bf x}_u \frac{f_{{\bf V}|{\bf X}}( {\bf v}_u^t| {\bf x}_u;K)f_{\bf X}({\bf x}_u;P_{\rm{I}},K)}{\int_{\mathcal{R}^{2^K}}
f_{{\bf V}|{\bf X}}( {\bf v}_u^t| {\bf z} ;K)f_{\bf X}({\bf z};P_{\rm{I}},K) {\rm d}  {\bf{z}} }  {\rm d}  {\bf{x}}_u \\ \label{uobn12}
& =  \sqrt{P_{\rm{I}}} \sum_{k=1}^{2^K}  \frac{f_{{\bf V}|{\bf X}}( {\bf v}_u^t| \bar{\bf e}_k ;K)}{\sum_{i=m}^{2^K}
f_{{\bf V}|{\bf X}}( {\bf v}_u^t| \bar{\bf e}_m ;K)} {\bf e}_k,
\end{align}
where $f_{{\bf V}|{\bf X}}( {\bf v}_u^t| {\bf x} ;K)$ $\big{(}f_{{\bf X}|{\bf V}}( {\bf x}| {\bf v}_u^t ;K)\big{)}$ denotes the conditional \ac{pdf} of the random vector ${\bf v}_u^t$ (${\bf x}$) given random vector ${\bf x}$ (${\bf v}_u^t$), and $\bar{\bf e}_k \triangleq \sqrt{P}_{\rm I}  {\bf e}_k$. The last equality in \eqref{uobn12} is obtained by replacing \eqref{qpdf} in the third equality and then using $\int_{\mathcal{R}^{2^K}} \delta ({\bf{x}}-{\bf{x}}_0) {\rm d}  {\bf{x}} =1$.
Moreover,
due to the decoupling effect, we can write ${\bf v}_u^t= {\bf{x}}_u + {\cal{N}}(0,\rho_t^2 {\bf I})$;
thus, we have
\begin{align}\label{89iop}
f_{{\bf V}|{\bf X}}( {\bf v}_u^t| {\bf x}_u ;K) =  \frac{1}{\eta} \exp\bigg{\{}\frac{-\|{\bf v}_u^t-{\bf x}_u\|^2}{2\rho_t^2}\bigg{\}},
\end{align}
where $\eta \triangleq (2\pi\rho_t^2)^{2^{K-1}}$.
Finally, by substituting \eqref{89iop} into \eqref{uobn12}, we obtain \eqref{eq25nv}.
Similarly, for the RBS-AMP MUD algorithm, we can write
\begin{align}\label{eq 57}
&\hat{x}_i^{t+1}  = \eta({{v_i^t}}) = \mathbb{E}\big{\{}X_i|X_i+\sigma_t Z=v_i^t \big{\}}\\ \nonumber
& = \int_{-\infty}^{+\infty} x_i f_{{ X}|{ V}}\big{(} x_i| {v}_i^t ;K) \ {\rm d} x_i \\ \nonumber
& =\frac{p \sqrt{P_{\rm I}}f_{{ V}|{ X}}( {v}_i^t| x_i=\sqrt{P_{\rm{I}}} ;K)}{pf_{{ V}|{ X}}( {v}_i^t| x_i=\sqrt{P_{\rm{I}}} ;K)+(1-p)f_{{ V}|{ X}}( {v}_i^t| x_i=0 ;K)},
\end{align}
where $p=1/2^K$. By substituting $f_{{ V}|{ X}}( {v}_i^t| x_i ;K) = \frac{1}{\sqrt{2\pi}} \exp\big{\{}{-(v_i^t-x_i)^2} / {2\sigma_t^2}\big{\}}$
into \eqref{eq 57}, and after some simplification, we obtain \eqref{Equation_34}.
\vspace{-1em}
\section{}\label{Appx_3}
At the $(t+1)$th iteration of the RBS-AMP algorithm, we have $\hat{\bf x}^{t+1} \triangleq [\hat{x}_1^{t+1} \ \hat{x}_2^{t+1} \ \dots \ \hat{x}_N^{t+1}]^{\rm T}$, where $x_i^{t+1}$ is given in \eqref{Equation_34}. By employing \eqref{Equation_2} and the fact that ${\bf w}$ is independent of ${\bf x}$ and ${\bf v}_u^t$ (due to decoupling), the average reconstruction error $\mathcal{E}$ at the $(t+1)$th iteration of the RBS-AMP algorithm can be written as
\begin{align}\label{Equation_39}
{\mathcal{E}}^{t+1}&=\frac{1}{M}\mathbb{E}\big{\{}({\bf{y}}-{\bf{C}}\hat{\bf x}^{t+1})^{\rm T} ({\bf{y}}-{\bf{C}}\hat{\bf x}^{t+1})\big{\}} \\ \nonumber
&=\frac{1}{M}\mathbb{E}\big{\{}{\bf{w}}^{\rm T}{\bf{w}}\}+\frac{1}{M}\mathbb{E}\{({\bf{x}}-\hat{\bf{x}}^{t+1})^{\rm T}{\bf{C}}^{\rm T}{\bf{C}}({\bf{x}}-\hat{\bf{x}}^{t+1})\big{\}}.
\end{align}
The first term in \eqref{Equation_39} denotes the noise power $\sigma_{\rm w}^2$, and the second term is the  interference  at the $(t+1)$th iteration, i.e., $I^{t+1}$.
By using \eqref{Equation_23}, for $N \rightarrow \infty$, we can write
\begin{align} \nonumber
&\frac{1}{N}\mathbb{E}\big{\{}({\bf{x}}-\hat{\bf{x}}^{t+1})^{\rm T}({\bf{x}}-\hat{\bf{x}}^{t+1})\big{\}}=\frac{1}{N} ({\bf{x}}-\hat{\bf{x}}^{t+1})^{\rm T}({\bf{x}}-\hat{\bf{x}}^{t+1})
\\ \label{uiuio232}
&=\mathbb{E}\big{\{}(\eta(X+\sigma_t Z)-X)^2\}= (\sigma^2_{t+1} - \sigma_{\rm w}^2)\beta,
\end{align}
where $\beta= M/N$ and
$\sigma^2_{t+1}$
is given in \eqref{Equation_18} by replacing $U$ with $X$.
Moreover, for $N\rightarrow \infty$ and $M\rightarrow \infty$, we have ${\bf{C}}^{\rm T} {\bf{C}}={\bf I}$.
Considering this equality and \eqref{uiuio232}, we can write
\begin{align}\label{eq349}
I^{t+1}&=\frac{1}{M}\mathbb{E}{\big \{}({\bf{x}}-\hat{\bf{x}}^{t+1})^{\rm T}{\bf{C}}^{\rm T}{\bf{C}}({\bf{x}}-\hat{\bf{x}}^{t+1})\big{\}}
\\ \nonumber
&=\frac{1}{M}({\bf{x}}-\hat{\bf{x}}^{t+1})^{\rm T}({\bf{x}}-\hat{\bf{x}}^{t+1})
=(\sigma^2_{t+1} - \sigma_{\rm w}^2).
\end{align}
By substituting \eqref{eq349} into \eqref{mue}, we obtain the  MUE of the  RBS-AMP MUD algorithm at the $(t+1)$th iteration as in \eqref{7883}.
\vspace{-1em}
\section{}\label{Appx_1}
For the proof, we need to show that the sample mean estimator in \eqref{Equation_29} is consistent \cite{stuart1963advanced}.
Let us define random variable
$
Q_t \triangleq \big{\|} \tilde{\bm \eta}_t\big{(} {\bf{U}}+\rho_t {\bf Z} \big{)}-{\bf U}\big{\|}^2$
with finite mean $\theta_t \triangleq \mathbb{E}\{Q_t\}$. The sample mean estimate of $\theta_t$ is given by
$\hat{\theta}_{t,L} =\frac{1}{L} \sum_{l=1}^{L} q_{t,l}=\frac{1}{L} \sum_{u=1}^{L}\|\hat{\bf u}_u^{t+1}-{\bf u}_u\|^2,
$
where $q_{t,l}$, $l=1,2,\dots,L$, are \ac{iid} drawn from random variable $Q_t$, and the second inequality comes from the decoupling behavior in \eqref{ddnlawqm} and \eqref{SEV} and the denoising operation in \eqref{Equation_6}.
By using Markov's inequality, we can write
\begin{align}\label{markov}
P\big{(}|\hat{\theta}_{t,L}-\theta_t|>\epsilon_t\big{)}
\leq \frac{\mathbb{E}\{|\hat{\theta}_{t,L}-\theta_t|\}}{\epsilon_t}.
\end{align}
For the sample mean estimator $\hat{\theta}_{t,L}$, we have
\begin{align}\label{total_var}
|\hat{\theta}_{t,L}-\theta_t|^2 = \Big{|}\frac{1}{L} \sum_{l=1}^{L} q_{t,l}-\theta_t \Big{|}^2 \leq
\frac{1}{L^2} \sum_{l=1}^{L} |q_{t,l}-\theta_t|^2.
\end{align}
By applying the statistical expectation to both sides of \eqref{total_var}, and then using \eqref{markov}, we obtain \cite{liu2018sparse}
\begin{align}\label{tywop}
P(|\hat{\theta}_{t,L}-\theta_t|^2>\epsilon_t^2) \leq \frac{\sigma_{Q_t}^2}{L\epsilon_t^2},
\end{align}
where $\mathbb{E}\{(Q_t-\theta_t)^2\}=\sigma_{Q_t}^2=\mathbb{E}\{|q_{t,l}-\theta_t|^2\}$.
The right-hand side of \eqref{tywop} goes to zero as $L\rightarrow \infty$. Hence,
for any small positive values $\epsilon_t$, we can write
\begin{align}
\mathop {\lim }\limits_{L \to \infty } P{\bigg (}\Big{|}\frac{1}{L} \sum_{u=1}^{L}\|\hat{\bf u}_u^{t+1}-{\bf u}_u\|^2-\theta_t\Big{|}^2 \geq \epsilon_t^2 \Big{)}=0.
\end{align}
\vspace{-1.8em}
\bibliographystyle{IEEEtran}
\balance
\bibliography{IEEEabrv,Ref}

\begin{thebibliography}{10}
\providecommand{\url}[1]{#1}
\csname url@samestyle\endcsname
\providecommand{\newblock}{\relax}
\providecommand{\bibinfo}[2]{#2}
\providecommand{\BIBentrySTDinterwordspacing}{\spaceskip=0pt\relax}
\providecommand{\BIBentryALTinterwordstretchfactor}{4}
\providecommand{\BIBentryALTinterwordspacing}{\spaceskip=\fontdimen2\font plus
\BIBentryALTinterwordstretchfactor\fontdimen3\font minus
  \fontdimen4\font\relax}
\providecommand{\BIBforeignlanguage}[2]{{%
\expandafter\ifx\csname l@#1\endcsname\relax
\typeout{** WARNING: IEEEtran.bst: No hyphenation pattern has been}%
\typeout{** loaded for the language `#1'. Using the pattern for}%
\typeout{** the default language instead.}%
\else
\language=\csname l@#1\endcsname
\fi
#2}}
\providecommand{\BIBdecl}{\relax}
\BIBdecl

\bibitem{mohammadkarimi2018signature}
M.~Mohammadkarimi, M.~A. Raza, and O.~A. Dobre, ``Signature-based nonorthogonal
  massive multiple access for future wireless networks: Uplink massive
  connectivity for machine-type communications,'' \emph{{IEEE} Veh. Technol.
  Mag.}, vol.~13, no.~4, pp. 40--50, Oct. 2018.

\bibitem{mohammadkarimi2021massive}
M.~Mohammadkarimi, O.~A. Dobre, and M.~Z. Win, ``Massive uncoordinated multiple
  access for beyond {5G},'' \emph{{IEEE} Trans. Wireless Commun.}, 2021.

\bibitem{shiv2022learning}
U.~S. Shiv, S.~Bhashyam, C.~R. Srivatsa, and C.~R. Murthy, ``Learning-based
  sparse recovery for joint activity detection and channel estimation in
  massive random access systems,'' \emph{{IEEE} Wireless Commun. Lett.},
  vol.~11, no.~11, pp. 2295--2299, 2022.

\bibitem{liu2018sparse}
L.~Liu, E.~G. Larsson, W.~Yu, P.~Popovski, C.~Stefanovic, and E.~De~Carvalho,
  ``Sparse signal processing for grant-free massive connectivity: A future
  paradigm for random access protocols in the internet of things,''
  \emph{{IEEE} Signal Process. Mag.}, vol.~35, no.~5, pp. 88--99, 2018.

\bibitem{di2020detection}
R.~B. Di~Renna, C.~Bockelmann, R.~C. de~Lamare, and A.~Dekorsy, ``Detection
  techniques for massive machine-type communications: Challenges and
  solutions,'' \emph{IEEE Access}, vol.~8, pp. 180\,928--180\,954, 2020.

\bibitem{zadik2019improved}
I.~Zadik, Y.~Polyanskiy, and C.~Thrampoulidis, ``Improved bounds on {Gaussian}
  {MAC} and sparse regression via {Gaussian} inequalities,'' in \emph{proc.
  ISIT}, 2019, pp. 430--434.

\bibitem{choi2019noma}
J.~Choi, ``{NOMA}-based compressive random access using {Gaussian} spreading,''
  \emph{{IEEE} Trans. Commun.}, vol.~67, no.~7, pp. 5167--5177, Mar. 2019.

\bibitem{muller2021massive}
R.~R. M{\"u}ller, ``Massive gaussian multiple-access by random coding with soft
  interference cancellation,'' in \emph{proc. IEEE ICC}, 2021, pp. 1--6.

\bibitem{mayer2015bayesian}
M.~Mayer and N.~Goertz, ``Bayesian optimal approximate message passing to
  recover structured sparse signals,'' \emph{arXiv preprint arXiv:1508.01104},
  2015.

\bibitem{ma2019approximate}
A.~Ma, Y.~Zhou, C.~Rush, D.~Baron, and D.~Needell, ``An approximate message
  passing framework for side information,'' vol.~67, no.~7, pp. 1875--1888,
  Feb. 2019.

\bibitem{rangan2011generalized}
S.~Rangan, ``Generalized approximate message passing for estimation with random
  linear mixing,'' in \emph{Proc. ISIT}, 2011, pp. 2168--2172.

\bibitem{donoho2009message}
D.~L. Donoho, A.~Maleki, and A.~Montanari, ``Message-passing algorithms for
  compressed sensing,'' \emph{Proceedings of the National Academy of Sciences},
  vol. 106, no.~45, pp. 18\,914--18\,919, 2009.

\bibitem{berthier2020state}
R.~Berthier, A.~Montanari, and P.-M. Nguyen, ``State evolution for approximate
  message passing with non-separable functions,'' \emph{Information and
  Inference: A Journal of the IMA}, vol.~9, no.~1, pp. 33--79, Jan. 2019.

\bibitem{rush2018finite}
C.~Rush and R.~Venkataramanan, ``Finite sample analysis of approximate message
  passing algorithms,'' \emph{{IEEE} Trans. Inf. Theory}, vol.~64, no.~11, pp.
  7264--7286, Mar. 2018.

\bibitem{bayati2011dynamics}
M.~Bayati and A.~Montanari, ``The dynamics of message passing on dense graphs,
  with applications to compressed sensing,'' \emph{{IEEE} Trans. Inf. Theory},
  vol.~57, no.~2, pp. 764--785, Jan. 2011.

\bibitem{stuart1963advanced}
M.~G. Kendall and A.~Stuart, \emph{The advanced theory of statistics Griffin},
  1961.

\end{thebibliography}

\end{document}